\documentclass{article}
\usepackage{times}
\usepackage{graphicx} 
\usepackage{amsmath, amssymb, amsthm}
\usepackage[top=1in, bottom=1in, left=1.5in, right=1.5in]{geometry}
\usepackage{mathrsfs}
\usepackage{enumerate}

\usepackage{apacite}
\let\cite\shortcite
\let\citeA\shortciteA

\usepackage{algorithm}
\usepackage{algorithmic}
\usepackage{authblk}

\newcommand{\com}[2]{\left(\begin{array}{c}#1 \\ #2\end{array}\right)}
\newcommand{\lb}{\left(}
\newcommand{\rb}{\right)}
\newcommand{\eps}{\epsilon}

\newcommand{\E}{\mathbb{E}}

\newcommand{\mc}{\mathcal}

\newcommand{\bigabs}[1]{\bigg|#1\bigg|}
\newcommand{\nullset}{\mathcal{H}_{0}}

\newcommand{\fdphat}{\widehat{\mathrm{FDP}}}
\newcommand{\fdp}{\mathrm{FDP}}
\newcommand{\fdr}{\mathrm{FDR}}
\newcommand{\cS}{\mathcal{S}}

\newtheorem{proposition}{Proposition}
\newtheorem{lemma}{Lemma}
\newtheorem{theorem}{Theorem}

\theoremstyle{definition}
\newtheorem{definition}{Definition}

\newtheorem{remark}{Remark}

\title{Power of Ordered Hypothesis Testing}
\author{Lihua Lei\thanks{lihua.lei@berkeley.edu}\,\,}
\author{William Fithian\thanks{wfithian@berkeley.edu} }
\affil{Department of Statistics, University of California, Berkeley}

\begin{document} 

\maketitle

\newcommand{\WFcomment}[1]{{\color{red}{(WF: \bf \sc #1) }}}
\newcommand{\LLcomment}[1]{{\color{blue}{(LL: \bf \sc #1) }}}
\allowdisplaybreaks

\begin{abstract} 
  Ordered testing procedures are multiple testing procedures that exploit a pre-specified ordering of the null hypotheses, from most to least promising. We analyze and compare the power of several recent proposals using the asymptotic framework of \citeA{Li15}. While accumulation tests including ForwardStop can be quite powerful when the ordering is very informative, they are asymptotically powerless when the ordering is weaker. By contrast, Selective SeqStep, proposed by \citeA{barber15}, is much less sensitive to the quality of the ordering. We compare the power of these procedures in different r\'{e}gimes, concluding that Selective SeqStep dominates accumulation tests if either the ordering is weak or non-null hypotheses are sparse or weak. Motivated by our asymptotic analysis, we derive an improved version of Selective SeqStep which we call Adaptive SeqStep, analogous to Storey's improvement on the Benjamini-Hochberg procedure. We compare these methods using the GEOQuery data set analyzed by~\cite{Li15} and find Adaptive SeqStep has favorable performance for both good and bad prior orderings.
\end{abstract} 

\section{Introduction}

Since the invention of the Benjamini--Hochberg (BH) procedure \cite{bh95}, control of the false discovery rate (FDR) has gained widespread adoption as a reasonable measure of error in multiple hypothesis testing problems. In a typical setup, we observe a sequence of p-values $p_1,\ldots,p_n$ corresponding to null hypotheses $H_1, \ldots, H_n$, then apply some procedure to reject a subset of them. If we make $R$ total rejections (also called ``discoveries'') of which $V$ are true nulls (false discoveries), then the false discovery proportion (FDP) and false discovery rate (FDR) are defined respectively as
\[
\mathrm{FDP} = \frac{V}{R\vee 1},\quad \mathrm{FDR} = \E\; \mathrm{FDP}.
\]
Let $\cS = \{i:\; H_i \text{ is rejected}\}$ and $\nullset = \{i:\; H_i \text{ is true}\}$, so that $R=|\cS|$ and $V=|\cS \cap \nullset|$.

We can classify testing problems into three types: batch testing, ordered testing and online testing.  In batch testing, the ordering of hypotheses is irrelevant. The BH procedure and its many variants \cite{ben97, ben06, storey02, gen06} have been shown effective in this setting both in finite samples and asymptotically \cite{gen02, storey02, storey04, fer06}.

By contrast, in ordered testing, the ordering of hypotheses encodes prior information, typically telling us which hypotheses are most ``promising'' (i.e., most likely to be discoveries). For example, in genomic association studies, biologists could have prior knowledge about which genes are more likely to be associated with a disease of interest, and use this knowledge to concentrate statistical power on the more promising genes. Because prior information of this type is quite prevalent in scientific research, procedures that exploit it are attractive. Alternatively, the ordering may arise from the mathematical structure of the problem. For example, the co-integration test \cite{engle87}, which is widely used in macro-economics, involves testing $H_{j}: \text{rank}(A)\le j$ where $A$ is a coefficient matrix. Because the hypotheses are nested, it makes no sense to accept $H_j$ and reject $H_{j+1}$. Other examples include sequential goodness-of-fit testing for the LASSO and other forward selection procedures such as \citeA{lockhart14, kozbur15, fithian15}, which test $H_{k}: \mathcal{M}^{*}\subset \mathcal{M}_{k - 1}$ where $\mathcal{M}^{*}$ is the true model and $\mathcal{M}_{k - 1}$ is the model selected in $(k - 1)$-th step. Section~\ref{sec:AdaptiveSeqStep} reviews methods for ordered testing include ForwardStop \cite{gsell15}, Accumulation Tests \cite{Li15}, SeqStep and Selective SeqStep \cite{barber15}.

Finally, in online testing, the ordering of hypotheses does not necessarily encode prior knowledge; rather, it imposes a constraint on the selection procedure, requiring that we decide to accept or reject $H_i$ before seeing data for later hypotheses. Online procedures include $\alpha$-investing \cite{foster08}, generalized $\alpha$-investing \cite{aha14}, LOND and LORD \cite{jav15}. We will not address the online setting here.

In Section~\ref{sec:AdaptiveSeqStep} we summarize existing ordered testing procedures and propose a new procedure, Adaptive SeqStep (AS), generalizing Selective SeqStep (SS). Our motivation is analogous to \citeA{storey02}'s improvement on the BH procedure. In Section~\ref{sec:asymptotic}, we introduce the varying coefficient two-groups (VCT) model and derive an explicit formula for asymptotic power of AS and SS under this model, comparing it to analogous results obtained by~\citeA{Li15} under similar asymptotic assumptions. Section~\ref{sec:powerComp} presents a detailed comparison of the asymptotic power of AS, SS, and accumulation tests (AT) under various regimes. In Section~\ref{sec:parameters}, we discuss selection of parameters and evaluate the finite-sample performance by simulation. In Section~\ref{sec:dosage}, we re-analyze the dosage response data from~\citeA{Li15}, illustrating the predictions of our theory in real data. Section~\ref{sec:conclusions} concludes.

\section{Ordered Testing and Adaptive SeqStep}\label{sec:AdaptiveSeqStep}

Let $\pi_0$ denote the fraction of null p-values. Unless otherwise stated, we assume that null p-values are independent of the non-null p-values, and are drawn i.i.d. from the uniform distribution $U[0,1]$.

We now summarize several batch testing and ordered testing procedures and relate them to each other. For all of the procedures discussed below, the set of discoveries is of the form $\cS(s,k) = \{i\leq k:\; p_i \leq s\}$: all p-values below some threshold $s$, which arrive before some stopping index $k$. Similarly $R(s,k), V(s,k),$ and $\fdp(s,k)$ denote the resulting values of $V,R,$ and FDP if we select $\cS(s,k)$. 

Moreover, each method operates by defining some estimator of $\fdp(s,k)$, then maximizing the number of rejections $R(s,k)=|\cS(s,k)|$ subject to a constraint that $\fdphat(s,k) \leq q$, the target $\fdr$ control level. For example, the BH procedure rejects all $H_i$ with $p_i\leq \hat{s}_{BH} = \max \{s:\, s \leq q R(s,n)/n\}$, and may be formulated as
\begin{align}
\max_{s\in [0,1]} R(s,n) &\quad \text{s.t.}\,\, \fdphat_{BH}(s)\le q;\label{eq:bhproblem}\\
\nonumber
\fdphat_{BH}(s) &= \frac{ns}{\sum_{i=1}^{n}I(p_{i}\le s) \vee 1} = \frac{\frac{1}{\pi_0}\E V(s,n)}{R(s,n)\vee 1}.
\end{align}
\citeA{bh95} show that $\mathrm{FDR}_{BH} \le \pi_{0}q$. The procedure is very conservative when $\pi_{0}$ is small because $\fdphat_{BH}(s)$ overestimates the true FDP. If $\pi_{0}$ were known, we could reduce $\fdphat_{BH}(s)$ by a factor $\pi_{0}$, obtaining a more liberal threshold $s$ (and therefore more rejections) while still controlling the FDR at level $q$. 

In most problems, $\pi_0$ is unknown. \citeA{storey04} propose an estimator based on counting the number of p-values {\em above} some fixed threshold $\lambda \in (0,1)$: 
\[\hat{\pi}_{0}(\lambda) = \frac{1 + \sum_{i=1}^{n}I(p_{i} > \lambda)}{n(1 - \lambda)} = \frac{1 + A(\lambda,n)}{n(1-\lambda)},\]
where $A(\lambda,k) = k - R(\lambda,k) = \sum_{i=1}^{k}I(p_{i} > \lambda)$ counts p-values exceeding the threshold. The logic is that, for high enough $\lambda$, the count $A(\lambda,n)$ will exclude most non-null p-values (commonly $\lambda=0.5$). The Storey-BH (SBH) procedure then modifies~\eqref{eq:bhproblem}, solving instead
\begin{align*}
\max_{s\in [0,\lambda]} R(s,n) &\quad \text{s.t.}\,\, \fdphat_{SBH}(s;\lambda)\le q;\\
\fdphat_{SBH}(s;\lambda) &= \hat\pi_0(\lambda)\, \fdphat_{BH}(s)\\
&= \frac{s}{1 - \lambda}\,\cdot\,\frac{1 + A(\lambda,n)}{R(s,n)\vee 1},
\end{align*}
\citeA{storey04} show that 
\[\mathrm{FDR}_{SBH} \le (1 - \lambda^{|\nullset|})q,\]
which can be much closer to $q$ than $\pi_{0}q$.

In ordered testing procedures, the choice variable is not the threshold $s$ but rather the stopping index $k$. For example, Selective SeqStep (SS) \cite{barber15} rejects all hypotheses $H_i$ with $p_i\le s$ and $i\le\hat{k}_{SS}$ where
\[\hat{k}_{SS} = \max_{k\le n}\left\{k: \frac{1 + \sum_{i=1}^{k}I(p_{i} > s)}{\sum_{i=1}^{k}I(p_{i}\le s)\vee 1}\le \frac{1 - s}{s}q\right\},\]
for a given $s\in (0, 1)$. This can be reformulated as
\begin{align*}
\max_{k\in \{0,\ldots,n\}} R(s,k) &\quad \text{s.t.}\,\, \fdphat_{SS}(k; s)\le q;\\
  \fdphat_{SS}(k; s) &= \frac{s}{1 - s}\,\cdot\,\frac{1 + A(s,k)}{R(s,k) \vee 1}.
\end{align*}
The close resemblance between $\fdphat_{SS}(k; s)$ and $\fdphat_{SBH}(s; \lambda)$ suggests writing $\fdphat_{SS}$ as 
\[
\fdphat_{SS}(k; s) = \hat\pi_0(s, k)\,\fdphat_{BH}(s, k),
\]
where the second argument $k$ indicates evaluation on only the first $k$ p-values. 

If the threshold $s$ is low, then $A(s, k)$ may include many non-null p-values, leading to an upwardly-biased estimate of $\phi_0$. This observation motivates introducing an additional parameter to improve the procedure, analogous to the improvement of SBH over BH. Defining
\[
\fdphat_{AS}(k; s, \lambda) = \frac{s}{1 - \lambda}\,\cdot\,\frac{1 + A(\lambda,k)}{R(s,k) \vee 1},
\]
we arrive at our proposal, which we call {\em Adaptive SeqStep} (AS): for some $0 \leq s \leq \lambda \leq 1$, reject all hypotheses with $p_i \le s$ and $i\le \hat k_{AS}$, where
\begin{equation}\label{eq:kas}
\hat{k}_{AS} = \max\{k: \fdphat_{AS}(k; s,\lambda)\le q\}.
\end{equation}
If $\lambda > s$ (say, $s=0.1$ and $\lambda=0.5$), then $\hat\pi_0(\lambda; k)$ may be much less upwardly biased than $\hat\pi_0(s; k)$, leading to a more powerful procedure. We investigate this power comparison in Sections~\ref{sec:asymptotic}--\ref{sec:powerComp}.

The following theorem shows that AS achieves exact FDR control in finite samples. 
\begin{theorem}\label{thm:as}
Let $\nullset\subset \{1, \ldots, n\}$ denote the set of nulls, and assume that $\{p_{i}:\; i\in \nullset\}$ are independent of $\{p_i:\; i\notin \nullset\}$, and i.i.d. with distribution function $F_0$ that stochatically dominates $U[0, 1]$. For $\hat{k}_{AS}$ defined as in (\ref{eq:kas}),
\[\mathrm{FDR}(\hat{k}_{AS}; s, \lambda) = \E \lb\frac{\sum_{i\in \nullset, i\le \hat{k}_{AS}}I(p_{i}\le s)}{\sum_{i\le \hat{k}_{AS}}I(p_{i} \le s)\vee 1}\rb \le q.\]
\end{theorem}
The proof of Theorem \ref{thm:as} is given in Appendix A. 

Another class of ordered testing procedures are {\em accumulation tests} (AT) \cite{Li15}, which include ForwardStop \cite{gsell15} and SeqStep \cite{barber15} as special cases. Accumulation tests estimate FDP via
\[\fdphat_{AT}(k) = \frac{1}{k}\sum_{i=1}^{k}h(p_{i}),\]
for some function $h\geq 0$ with $\int_{0}^{1}h(x)dx = 1$, and rejects all hypotheses $H_i$ with $i\le \hat{k}$ where
\[\hat{k} = \max\left\{k: \fdphat_{AT}(k)\le q\right\}.\]
ForwardStop corresponds to the case where $h(x) = -\log(1 - x)$ and SeqStep corresponds to the case where $h(x) = CI(x > 1 - 1/C)$ for some $C > 0$.

In terms of our framework, accumulation tests solve 
\[
  \max_{k\in \{0,\ldots,n\}} R(1,k) \quad \text{s.t.}\,\, \fdphat_{AT}(k)\le q.
\]
The main difference between AT and AS is that the former rejects all hypotheses before $\hat{k}$, while the latter rejects only those smaller than threshold $s$. This means that AT will have full power if $\hat{k} \to n$, while the power of AS or SS is at most the average probability that a non-null p-value is less than $s$. On the other hand, unless nearly all of the early hypotheses are non-null, AT is likely to stop very early, as we will explore in Section~\ref{sec:powerComp}.

\section{Asymptotic Power Calculation}\label{sec:asymptotic}
\subsection{Varying Coefficient Two-group (VCT) Model}
We now derive the asymptotic power of AS and SS under the following simple model:
\begin{definition}[Varying Coefficient Two-groups (VCT) Model]\label{def:VCT}
  An VCT$(F_{0}, F_{1}; \pi(\cdot))$ model is a sequence of independent p-values $p_{i}\in [0, 1]$ such that
\[p_{i}\sim \lb 1 - \pi\lb i/n \rb\rb F_{0} + \pi\lb i/n \rb F_{1}\]
for some distinct distributions $F_{0}$ and $F_{1}$ and a non-negative function $\pi(t): [0, 1]\rightarrow [0, 1]$. $F_{0}$ and $F_{1}$ are the null and non-null distributions and $\pi(t)$ is the local non-null probability for $k= nt$.
\end{definition}
For simplicity, we will take $F_0$ to be uniform. Following \citeA{gen06}, we also assume that $F_1$ is strictly concave, so the density $f_1$ of non-null p-values is strictly decreasing; in other words, smaller p-values imply stronger evidence against the null.

The cumulative non-null probability $\Pi(t)$ is
\[\Pi(t) = \frac{1}{t}\int_{0}^{t}\pi(s)ds.\]
The quantity $\Pi(t)$ is essential to our results. It can be regarded as the average proportion of non-null hypotheses in the first $nt$-hypotheses since
\[\frac{\#\{i\le nt: i\not\in \mathcal{H}_{0}\}}{nt} \approx \frac{1}{t}\int_{0}^{t}\pi(s)ds = \Pi(t).\]
Our setting is very similar to that of \citeA{Li15} except that they impose conditions on $\Pi(t)$ directly. Proposition 1 in Appendix B reveals the relation between the VCT model and the assumptions of \citeA{Li15}. 

\subsection{Asymptotic Power for AS and SS}\label{sec:pow_as}

Because the SS method is a special case of the AS method with $\lambda=s$, it is sufficient to analyze the general case of AS. Assuming a VCT model, for large $n$, we have 
\begin{align*}
&\fdphat_{AS}(\lfloor nt\rfloor) = \frac{s}{1 - \lambda}\cdot\frac{1 + A(\lambda, \lfloor nt\rfloor)}{R(s, \lfloor nt\rfloor)}\\[4pt]
\approx
&\frac{s}{1 - \lambda}\cdot \frac{(1 - \Pi(t))(1 - \lambda) + \Pi(t)(1 - F_{1}(\lambda))}{(1 - \Pi(t))s + \Pi(t)F_{1}(s)}\\[4pt]
=& \frac{1 + \Pi(t)\left(\frac{1 - F_{1}(\lambda)}{1-\lambda} - 1\right)}
{1 + \Pi(t)\left(\frac{F_{1}(s)}{s}-1\right)} \triangleq \fdp_{AS}^*(t).
\end{align*}
Because $F_{1}$ is strictly concave, we have $\frac{1 - F_{1}(\lambda)}{1 - \lambda} < 1 < \frac{F_1(s)}{s}$, and $\fdp_{AS}^*(t)$ is a strictly decreasing function of $\Pi(t)$. Thus, in the limit, $\fdphat_{AS}(k)$ is determined by the fraction of non-nulls $\Pi(t)$, with more non-nulls leading to a lower estimate of FDP. Setting $\fdp_{AS}^*(t)=q$ and solving for $\Pi(t)$, we obtain the critical non-null fraction
\begin{equation}\label{eq:chi_as}
\chi_{AS}(s, \lambda, q, F_1) = \frac{1 - q}{1 - \frac{1 - F_{1}(\lambda)}{1 - \lambda} + q\lb\frac{F_{1}(s)}{s} - 1\rb}.
\end{equation}

\begin{figure}[H]
\begin{center}
    \includegraphics[width = 0.7\columnwidth]{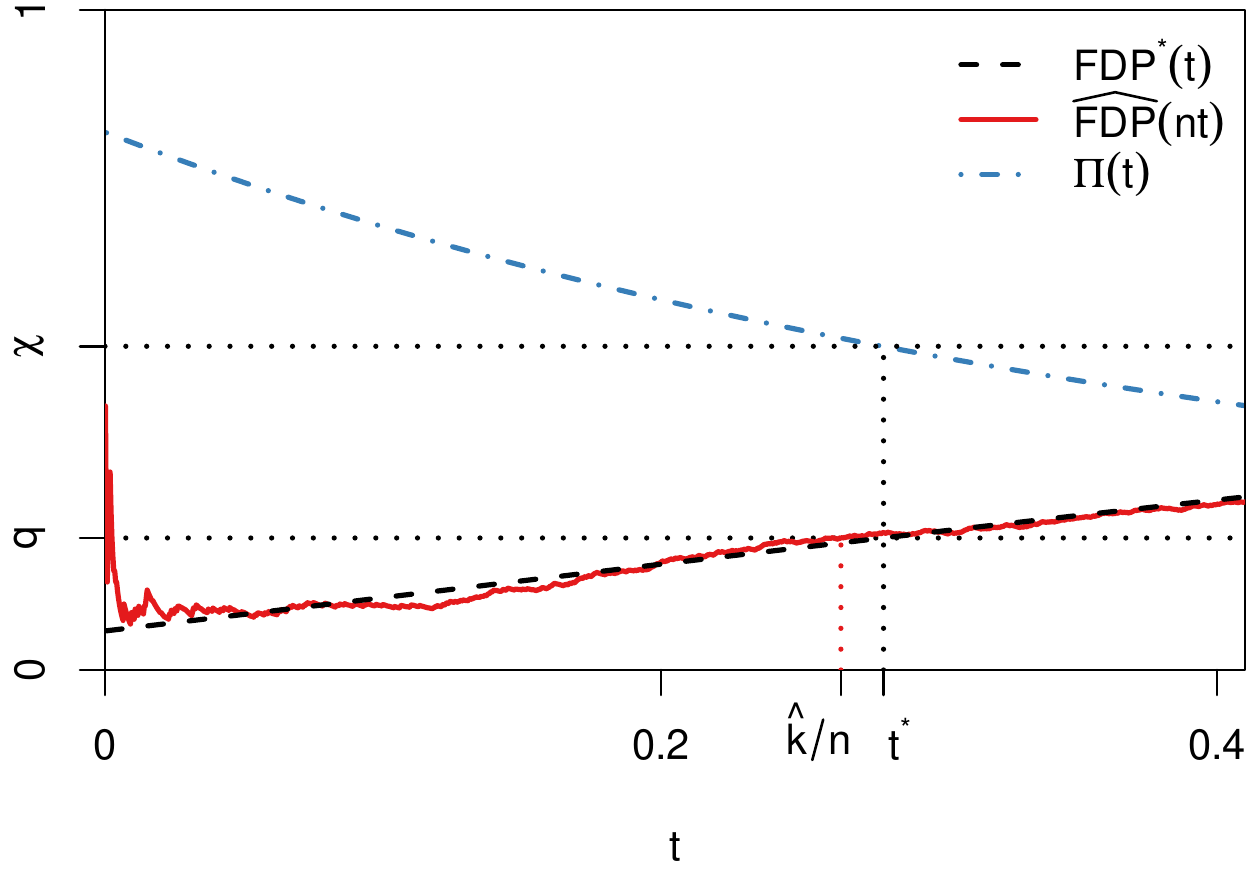}
    \caption{Illustration of the asymptotic behavior of AS. The broken curves show population limits for the simulation of Section~\ref{sec:powerComp}, with parameters set to $\lambda=0.5,s=q=\gamma=0.2,\mu=2,b=5$. At $t_{AS}^*$, $\Pi(t)=\chi_{AS}$, leading to $\fdp_{AS}^*(t)=q$. The red curve shows a realization of $\fdphat_{AS}(nt)$ with $n=3000$.}\label{fig:asymp}
  \end{center}
\end{figure}

If $\Pi(k/n) > \chi_{AS}$ then, with high probability, we will have $\fdphat_{AS}(k) \leq q$, implying $\hat{k}_{AS} \geq k$. The proportion $\hat{k}_{AS} / n$ of scanned hypotheses is approximately
\begin{equation}\label{eq:T_as}
t_{AS}^* = \max\{t: \Pi(t)\ge \chi_{AS}\},
\end{equation}
and the realized power is approximately
\begin{align}
\mathrm{Pow}_{AS} & = \frac{\#\{i\le \hat{k}: i\not\in \nullset, p_{i}\le s\}}{\#\{i\le n: i\not\in \nullset\}}\nonumber\\ 
& = \frac{\hat{k}}{n}\cdot\frac{\#\{i\le \hat{k}: i\not\in \nullset, p_{i}\le s\} / \hat{k}}{\#\{i\le n: i\not\in \nullset\}/ n}\nonumber\\
& \approx  F_1(s)\cdot \frac{t^*_{AS}\Pi(t^*_{AS})}{\Pi(1)}.\label{eq:power_as}
\end{align}
Theorem~\ref{thm:power_as} confirms our heuristic approximations.
\begin{theorem}\label{thm:power_as}
Consider a VCT model with
\begin{itemize}
\item $\Pi(t)$ is strictly decreasing and Lipschitz on $[0, 1]$ with $\Pi(1) > 0$;
\item $F_{0}$ is the uniform distribution on $[0, 1]$ and $f_{1} = F_{1}'$ is strictly decreasing on $[0, 1]$.
\end{itemize}
Then $\hat{k}_{AS} / n\stackrel{a.s.}{\rightarrow}t^*_{AS}$ and
\[\mathrm{Pow}_{AS} \stackrel{a.s.}{\rightarrow}F_{1}(s) \cdot\frac{t^*_{AS}\Pi(t^*_{AS})}{\Pi(1)} 
= F_{1}(s) \cdot\frac{\int_0^{t^*_{AS}}\pi(u)du}{\int_0^{1}\pi(u)du},\]
with $t_{AS}^*$ defined as in~\eqref{eq:T_as}.
\end{theorem}

Interpreting~(\ref{eq:T_as}--\ref{eq:power_as}), we see that if $\Pi(1)>\chi_{AS}$ then $t_{AS}^*=1$ and $\hat{k}_{AS}=n$ with high probability: all $p_i<s$ are rejected and power is roughly $F_1(s)$. Conversely, if $\chi_{AS} > \sup_{t\in[0,1]} \Pi(t)$ then $\hat{k}_{AS} = o_p(n)$ and the method is asymptotically powerless. Figure~\ref{fig:asymp} illustrates an intermediate case with $0<t_{AS}^*<1$.

From Theorem~\ref{thm:power_as} we see there are two ways to increase the asymptotic power: either increase $s$ (which we can do directly), or increase $t_{AS}^*$. To increase $t_{AS}^*$ we must {\em decrease} $\chi_{AS}(s,\lambda,q,F_1)$, which itself is increasing in $s$ and decreasing in $\lambda$. 

Increasing $\lambda$ always increases the asymptotic power. Because SS is a special case of AS, this implies that SS can always be improved by increasing $\lambda$ above $s$, yielding a less biased estimator of the null fraction. Note, however, that taking $\lambda \to 1$ is not practical in finite samples: we still need large enough $A(\lambda, k)$ for the estimator to be stable.

Increasing $s$ has an ambiguous effect on the asymptotic power. A smaller $s$ leads to a smaller $\chi_{AS}$, and therefore a larger stopping index $\hat{k}_{AS}$; however, it also applies a more stringent rejection threshold for hypotheses with $i\leq \hat{k}_{AS}$. By contrast, larger $s$ is more liberal for $i \leq \hat{k}_{AS}$ but tends to give smaller $\hat{k}_{AS}$. If $s$ is too large, $\chi$ could even exceed $\Pi(0)$, leading to a total loss of power. Small $s$ avoids this catastrophe: if $\lim_{s\to 0}F_1(s)/s = \infty$ then $\lim_{s\to 0}\chi_{AS}= 0$. This implies that we can always have nonzero power if we take $s$ small enough, but the power can never exceed $F_1(s)$ even if $\hat{k}_{AS}\approx n$. Intuitively, then, using a large $s$ is more aggressive, gambling that $\Pi$ is large enough to overcome the larger value of $\chi_{AS}$.

\subsection{Asymptotic Power for AT}
For AT, \citeA{Li15} prove that
\begin{equation}
  \label{eq:power_at}
  \mathrm{Pow}_{AT} \stackrel{a.s.}{\rightarrow} \frac{t^{*}_{AT}\Pi(t^{*}_{AT})}{\Pi(1)}
\end{equation}
where $t^{*}_{AT} = \max\{t: \Pi(t) \ge \chi_{AT}\}$, where
\begin{equation}\label{eq:chi_at_general}
\chi_{AT}(h, q, F_{1}) = \frac{1 - q}{1 - \nu},
\end{equation}
and $\nu = \E_{p\sim F_{1}}h(p)$. They also show that SeqStep, which uses $h(x) = CI(x > 1 - 1 / C): C \in (0, 1)$, is most powerful among all accumulation tests with $h$ bounded by $C$. Reparameterizing with $\lambda=1-1/C$, we can write $h(x) = \frac{1}{1 - \lambda}I(x > \lambda)$. Then, $\nu = \frac{1 - F_{1}(\lambda)}{1 - \lambda}$ and
\begin{equation}\label{eq:chi_at}
\chi_{AT} = \frac{1 - q}{1 - \frac{1 - F_{1}(\lambda)}{1 - \lambda}}.
\end{equation}
Comparing (\ref{eq:chi_at}) with (\ref{eq:chi_as}) and recalling that $F_{1}(s) > s$ by concavity, we see that $\chi_{AS} < \chi_{AT}$. Therefore, $t^{*}_{AS}\ge t^{*}_{AT}$, implying that AT will tend to stop earlier than AS. Even so, AT could be more powerful due to the extra factor $F_{1}(s)$ in (\ref{eq:power_as}) which is absent from (\ref{eq:power_at}). If $f_1(x)=F_1'(x)$, we have
\[\nu = \frac{\int_{0}^{1}h(p)f_{1}(p)dp}{\int_{0}^{1}h(p)dp}\ge \inf_{x\in [0, 1]}f_{1}(x),\]
where the last term equals $f_{1}(1)$ if $F_{1}$ is strictly concave. Thus, for any choice of $h$ (bounded or otherwise), we have $\chi_{AT} \ge \frac{1-q}{1-f_1(1)}$.

\section{Power Comparisons}\label{sec:powerComp}

In this section we analyze the results of Section~\ref{sec:asymptotic} to extract further information about when each of AS, SS, and AT performs better or worse, and how and when the choice of $s$ affects the performance of AS. There are three salient features of the VCT model to consider:
\begin{description}
\item[Signal density] $\Pi(1) = \int_0^1\pi(t)dt$ gives the expected total number of nulls. Note $\Pi(1)=1-\pi_0$.
\item[Signal strength] If the non-null p-values tend to be very small, we say the signals are strong.
\item[Quality of the ordering] If the prior information is very good then $\Pi(t)$ is steep, with $\Pi(0)=1$ in the limit of very good information; if the prior ordering is completely useless then $\Pi(t) = \Pi(1)$ for all $t$.
\end{description}

First, note that if signals are very strong, then most of the non-null p-values are close to 1. In that case,
\[
\frac{1-F_1(\lambda)}{1-\lambda} \approx 0 \;\;\Rightarrow\;\; \chi_{AS} \approx \frac{1 - q}{1 + q\lb\frac{F_{1}(s)}{s} - 1\rb},
\]
even for relatively small values of $\lambda$, possibly including $\lambda=s$. As a result, $\lambda$ plays a very small role in determining $\chi_{AS}$ and AS will behave similarly to SS. By contrast, if the signals are weaker, the difference is greater.

Second, if the ordering is very good, with $\Pi(0)\approx 1$ and $\Pi(t)$ correspondingly very steep, then we can afford to use a larger $s$ for the AS procedure without worrying that $\chi_{AS}>\Pi(0)$ (though we still cannot allow $\chi_{AS}$ to exceed 1). By contrast, if the ordering is poor and $\Pi(t)$ is very flat, then a small change in $s$ could move $\chi_{AS}$ from below $\Pi(1)$ (for which $\hat{k}_{AS}\approx n$) to above $\Pi(0)$ (for which $\hat{k}_{AS} = o_p(n)$), and so we are forced to be very cautious.

Finally, examining~\eqref{eq:chi_at_general}, we see that AT is highly aggressive compared to AS. Suppose $q=0.1$. Then, regardless of the choice of $h$, AT is powerless unless at least $90\%$ of the early hypotheses are non-null, requiring that either the signals are very dense or the ordering is very informative. In addition, the signals must be quite strong: even if $\Pi(0)=1$, AT is asymptotically powerless unless
\[
\nu = \E_{p\sim F_{1}}h(p) < q \ll 1 = \E_{p\sim F_0}h(p).
\]

\begin{figure}[ht]
  \begin{center}
    \includegraphics[width = 0.7\columnwidth]{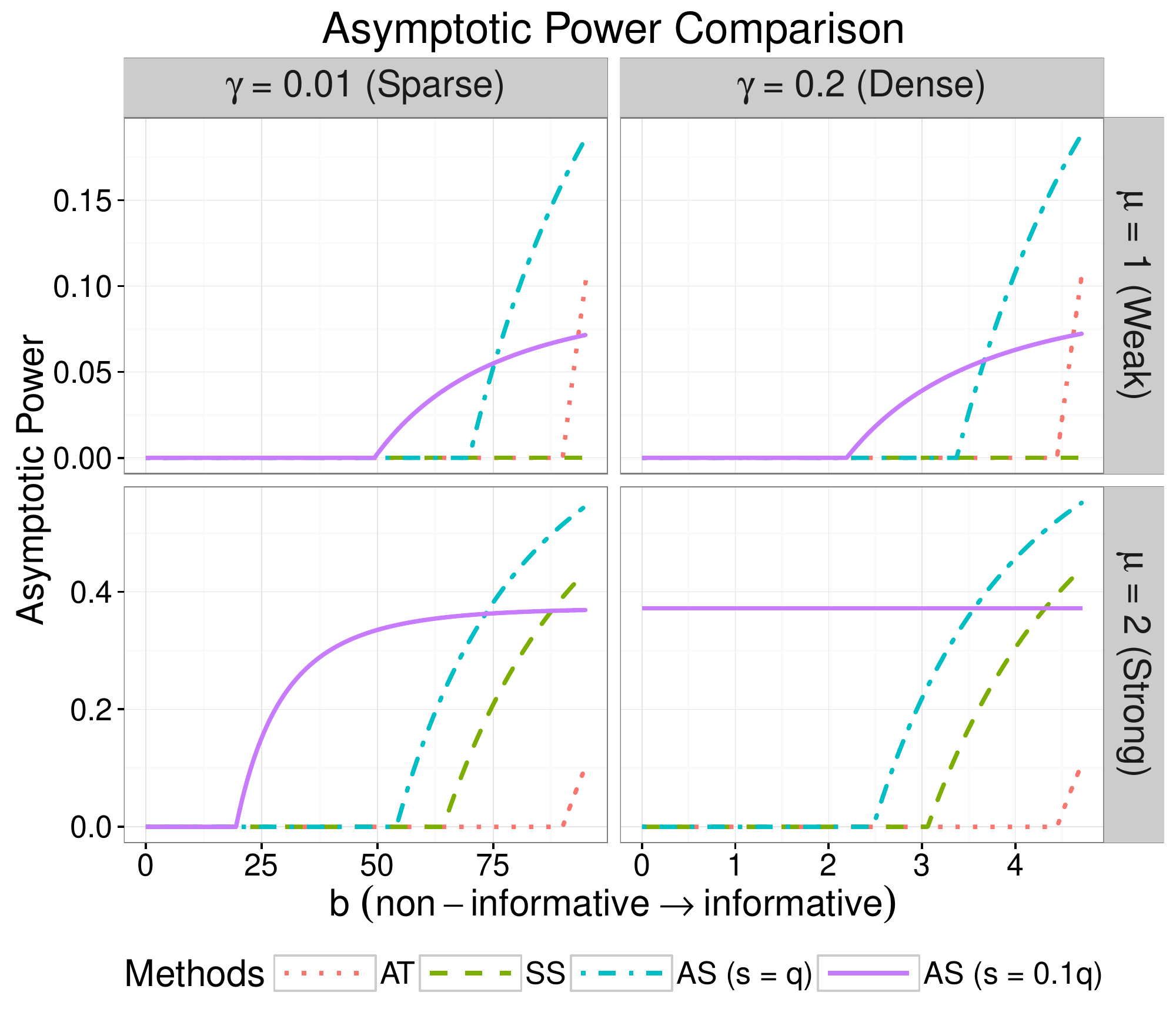}
   \end{center}
\caption{Asymptotic Power of AS (with $s = q$ and $s = 0.1q$), SS (with $s = q$) and AT (with $\nu = 0$) under four regimes: (sparse/weak) $\gamma = 0.01, \mu = 1$; (sparse/strong) $\gamma = 0.01, \mu = 2$; (dense/weak) $\gamma = 0.2, \mu = 1$; (dense/strong) $\gamma = 0.2, \mu = 2$. The x-axis measures $b$ and a larger $b$ corresponds to a more informative ordering.} \label{fig:power}
\end{figure}

\subsection{Numerical Results}

We now illustrate the above comparisons with a numerical example. We consider the VCT model where $F_{0}$ is uniform and $F_{1}$ is the distribution of one-tailed p-values from a normal test. That is, $p = \bar\Phi(z) = 1-\Phi(z)$ where $z\sim N(\mu,1)$ and $\Phi$ is the standard normal CDF. Thus,
\[F_{1}(x) = \bar\Phi(\bar\Phi^{-1}(x) - \mu),\]
with $\mu$ determining the signal strength.

For the local non-null density, we take
\[
\pi(t) = \gamma e^{-bt}\cdot\frac{b}{1-e^{-b}}, \quad \gamma\in (0,1), \;\;b> 0.
\]
The factor $b/(1-e^{-b})$ is a normalization constant guaranteeing $\Pi(1)=\int_0^1\pi(t)dt = \gamma$. Thus, $\gamma$ determines the signal density, while $b$ determines the quality of the prior ordering, with a larger $b$ corresponding to a better ordering and $b\to 0$ corresponding to a useless ordering. $b$ is implicitly upper-bounded by the constraint $\pi(0) = \gamma \cdot \frac{b}{1-e^{-b}} \leq 1$; let $b_{\max}$ denote the maximal value.

Figure~\ref{fig:power} shows the asymptotic power for four methods, all using $q=0.1$: AS with $s = q$ and $\lambda=0.5$, AS with $s = 0.1q$ and $\lambda=0.5$, SS with $s=q$, and AT. AT is not implemented with a specific $h$, but rather with  $\nu = 0$, giving the best possible power that any $h$ could achieve. Four regimes are shown corresponding to two levels each of $\mu$ and $\gamma$: $\mu = 1$ (weak signals) vs. $\mu = 2$ (strong signals), and $\gamma=0.01$ (sparse signals) vs. $\gamma = 0.2$ (dense signals). In each regime, we plot the asymptotic power of each method for $b\in(0,b_{\max}]$.

Unsurprisingly, all of the methods perform better with stronger, denser signals and better prior orderings, but their sensitivities to these parameters are quite different. Comparing the two AS methods, we see that smaller (less aggressive) $s$ makes the method less sensitive to the ordering quality: its power is usually positive, but it is outperformed by AS($s=q$) when the ordering is excellent. AT is even more aggressive than the other two, and is asymptotically powerless unless the ordering is excellent.

SS is dominated by AS($s=q$) in all cases, as predicted, but the improvement is less dramatic when the signals are strong; in that case $\frac{1-F_1(\lambda)}{1-\lambda} \approx 0.05$ and $\frac{1-F_1(s)}{1-s} \approx 0.26$ are both small compared to $1+q\left(\frac{F_1(s)}{s}-1\right) \approx 1.66$.

\section{Selection of Parameters}\label{sec:parameters}

\subsection{Selecting $\lambda$}
As explained in Section~\ref{sec:pow_as}, a large $\lambda$ reduces $\chi_{AS}$ and improves asymptotic power. However, in finite samples, the procedure will be unstable if $\lambda$ is too close to 1. One natural suggestion is to set $\lambda = 0.5$, analogous to Storey's suggestion for the Storey-BH procedure \cite{storey04}.

\subsection{Selecting $s$}
As discussed in Section~\ref{sec:pow_as}, $s$ has an ambiguous effect on the asymptotic power. The oracle choice $s^*$, which maximizes asymptotic power, is unknown in practice and depends on knowing parameters like $\Pi(t)$ and $F_1(x)$.
Although we could plug in estimators of the parameters $b$ and $\mu$, or simply choose the value of $s$ giving us the largest power on our data, the validity of such procedures would not be guaranteed by our results here.

In our view $s>q$ is intuitively unappealing because it would mean using a more liberal rejection cutoff than unadjusted marginal testing. We suggest $s=q$ as a heuristic, moderately aggressive default. This will give non-zero power as long as
\begin{equation}\label{eq:ruleofthumb}
\frac{F_{1}(q)}{1 - q} > \frac{1 - \Pi(0)}{\Pi(0)}.
\end{equation}
\eqref{eq:ruleofthumb} can be easily derived from (\ref{eq:chi_as}) and (\ref{eq:T_as}), provided $\lambda$ is close to 1 such that $\frac{1 - F_{1}(\lambda)}{1 - \lambda}\approx 0$, and is not too stringent. For example, if $q = 0.1$, $F_{1}(0.1) \geq 0.5$, then~\eqref{eq:ruleofthumb} holds provided $\Pi(0) > 0.64$. That is, if the non-nulls have reasonably strong signal and most of the early p-values are non-null, then $s=q$ is small enough. If we do not find these values of $F_1(0.1)$ and $\Pi(0)$ plausible, we can repeat this reasoning for smaller values of $s$ until we arrive at assumptions we do find plausible.

\subsection{Finite Sample Performance}
Now we evaluate the finite sample performances of the above two heuristics for $\lambda$ and $s$. Figure \ref{fig:finite} displays the distribution of realized power for AS using $\lambda=0.5$ vs. $\lambda=0.95$, and $s=q$ vs. $s=s^*$. We set $q = 0.1, \gamma = 0.2, \mu = 2, b = 3.65$, in which case $\Pi(0) = 0.75, \Pi(1) = 0.2$. Each panel shows power for $n = 100, 500, 1000,$ and $10,000$. For each setting we simulate 500 realizations of the fraction of all non-nulls that the method discovers. It is clear that large $\lambda$ is less stable especially when $n$ is small.

\begin{figure}[ht]
  \centering
  \includegraphics[width = 0.7\columnwidth, height = 0.5\columnwidth]{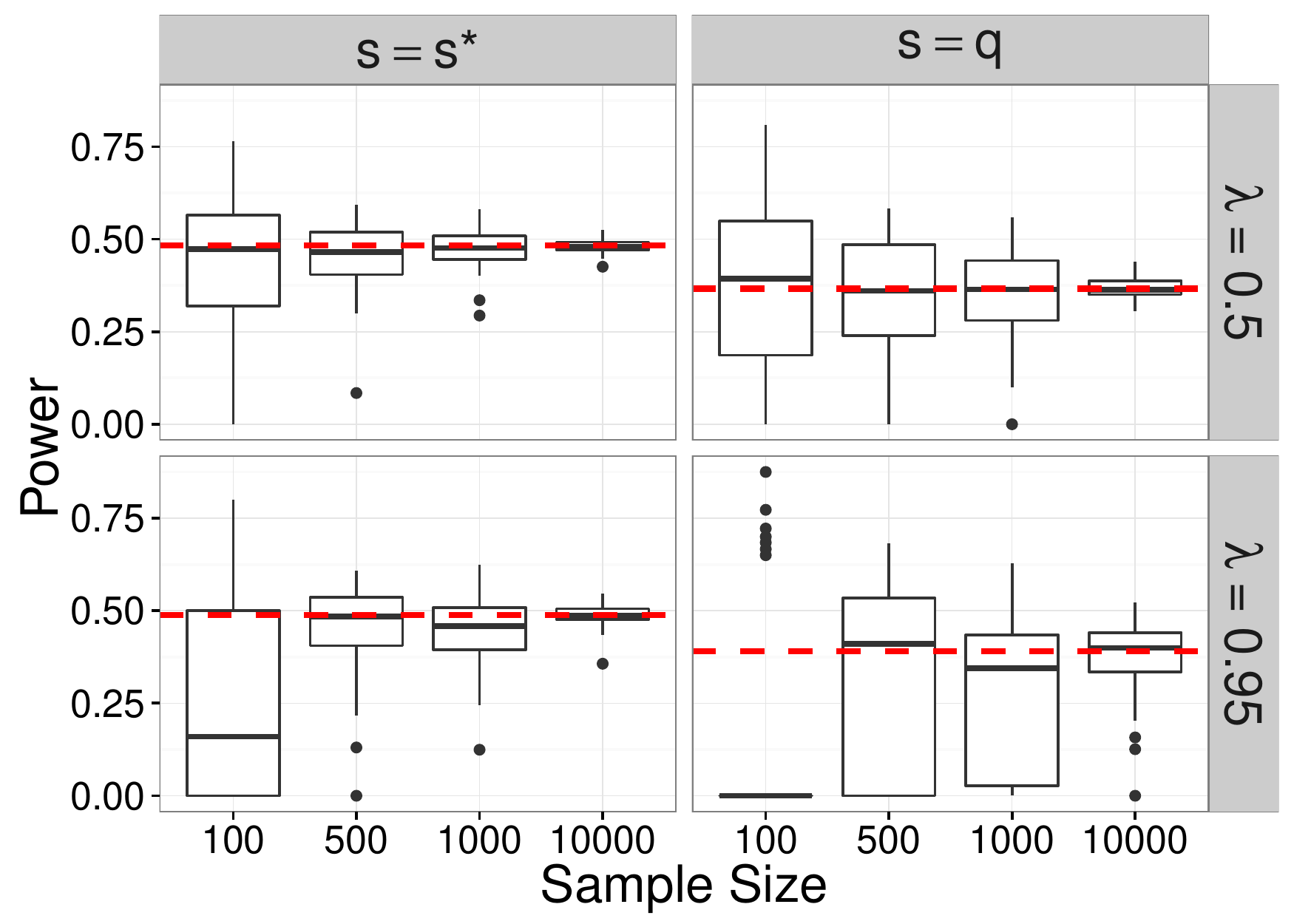}
  \caption{Finite-sample power using $s\in\{s^{*}, q\}$ and $\lambda\in \{0.5, 0.95\}$. Red dashed line corresponds to the asymptotic power.}\label{fig:finite}
\end{figure}

We see from Figure~\ref{fig:finite} that the performance of $\lambda = 0.5$ is more stable than that of $\lambda = 0.95$. On the other hand, the choice $s = q$ has a comparable power to the oracle approach. This justifies the simple choice $s = q$ as a moderately aggressive default choice for fairly strong signals and a good prior ordering; note however that if the signals were much weaker or the ordering much worse, $s=q$ could be powerless.

\section{Data Example: Dosage Response Data}\label{sec:dosage}

\citeA{Li15} analyzed the performance of several ordered testing methods using the GEOquery data of \citeA{davis2007geoquery}. In this section we reproduce and extend their analysis, adding the SS and AS methods as competitors. Where possible we have re-used the R code provided by \citeA{Li15} at their website.

The GEOquery data consist of gene expression measurements at 4 different dosage levels of an estrogen treatment for breast cancer, plus control (no dose). At each of the 5 dosage levels, the gene expression of $n=22,283$ genes is measured in 5 trials. The problem is to test whether each gene is differentially expressed in the lowest dosage versus the control condition, while using data from other dosage levels to obtain a prior ordering on the genes.

For each gene, \citeA{Li15} carry out a $t$-test comparing expression under the highest dose versus the expression under the lowest dose and control, pooled together. Let $\widetilde T_i$ denote the $t$-statistic and $\tilde p_i$ the p-value for gene $i$ using the high-dose data. Next, they compute one-sided permutation p-values $p_i$ comparing lowest dose to control, using the sign of $\widetilde T_i$ to determine which side. Finally, they order the p-values $p_1,\ldots,p_n$ according to the ordering of $\tilde p_1,\ldots, \tilde p_n$ and apply an ordered testing procedure. For a more detailed explanation of the experiment, see \citeA{Li15}. 

The top panel of Figure~\ref{fig:dosage} reproduces Figure~6 of \citeA{Li15} (with different axis limits), but including the SS and AS procedures analyzed here as competing methods. Both the HingeExp and AS methods perform quite well compared to the other methods, with SS coming in third place. In light of the foregoing theory, we can conclude that the high-dose data are doing an excellent job discriminating between null and non-null hypotheses --- for example, the HingeExp method rejects the first 600 hypotheses at the $q=0.1$ level, essentially implying that at least 540 of the first 600 genes in the ordering are truly non-null. The BH and Storey-BH methods, which are performed without any regard to the (highly informative) ordering, are unable to make any rejections.

For the lower panel of Figure~\ref{fig:dosage}, we repeat the same analysis, but with one change: instead of comparing with the highest dose to obtain $\widetilde T_i$, we instead compare with the second-lowest dose. This has the affect of attenuating the signal strength of $\widetilde T_i$, and thereby deteriorating the quality of the prior ordering. With a weaker ordering, all of the AT methods suffer major losses in power, so that the AS method is the clear winner, with SS in second place. As before, the BH and Storey methods have no power. This panel confirms the message of our theoretical analysis that AS and SS are more robust to weaker orderings.

\begin{figure}[ht]
  \centering
  \includegraphics[width = 0.7\columnwidth]{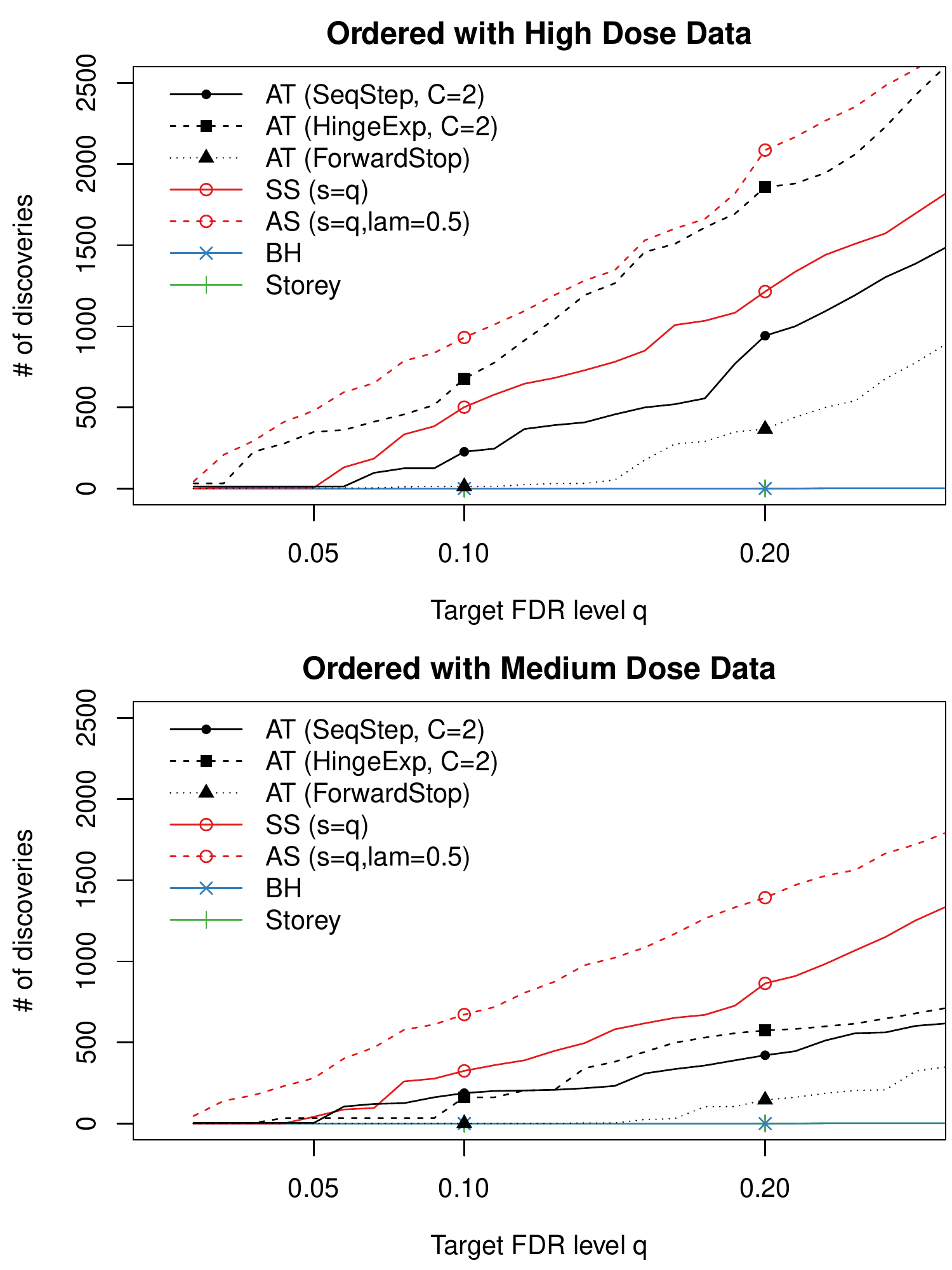}
  \caption{Power of AS, SS, and several AT methods on the dosage response data analyzed by \protect\citeA{Li15}.}\label{fig:dosage}
\end{figure}

\section{Conclusions and Future Directions}\label{sec:conclusions}
We have proposed Adaptive SeqStep (AS), which extends Selective SeqStep (SS) and improves on it in a manner analogous to Storey's improvement on the BH procedure. We have shown it controls FDR exactly in finite samples and analyzed its asymptotic power in detail, using the varying coefficient two-groups (VCT) model as a benchmark for comparing ordered testing procedures. For VCT models, we show that AS dominates SS asymptotically and outperforms AT except possibly in regimes with very good hypothesis ordering and very strong signals. Note that perfect ordering of hypotheses is implicit in the mathematical structure of many problems such as sequential goodness-of-fit testing; as a result AT could still be a suitable choice for these. Although we have proposed the heuristic $s=q$ for selecting $s$, it would be interesting to investigate whether there is a good way to estimate a good $s$ from the data.

A natural extension of AS is to allow $s$ and $\lambda$ to be different across the hypotheses. Intuitively, for those which have a higher chance to be non-null, we could use a more liberal threshold. Once the conditions for exact FDR control are established, we can derive the ``optimal'' $s$-sequence and $\lambda$-sequence under the asymptotic framework. We leave this as future work.

Another interesting direction is to compare AS and AT with BH-type methods. \citeA{gen02} has obtained the explicit formula for the power of BH procedure and it is not hard to obtain it in our more framework. The comparison should reveal more similarities and differences between these two genres.

Finally, AS is a natural fit for the ``multiple knockoffs'' extension of the knockoffs procedure suggested at the end of~\citeA{barber15}. Because the original knockoff procedure only produces 1-bit p-values, AS and SS are essentially equivalent, with $s=\lambda=0.5$ the most natural settings of those parameters. However, a multiple-knockoff procedure could yield p-values lying in $\{\frac{i}{k+1}: i= 1, 2, \ldots, k + 1\}$ by using $k$ knockoffs for each predictor variable. It would be interesting to see whether using the AS instead of SS procedure would give a meaningful improvement.

\section*{Acknowledgments}

We thank the anonymous reviewers for their helpful comments, which greatly improved this work.

\section*{Appendix A: Proof of Theorem \ref{thm:as}}
 Let $\mc{F}_{k}$ be the $\sigma$ field generated by all non-null pvalues as well as $\{I(p_{i}\le s), I(p_{i} > \lambda): i\ge k\}$. Then $\hat{k}$ is a stopping time with respect to the backward filtration $\mc{F}_{n}\subset \mc{F}_{n - 1}\subset \cdots \subset \mc{F}_{1}$. Recall that $V(s, k) = \sum_{i\le k, i\not\in\nullset}I(p_{i}\le s)$ and $R(s, k) = \sum_{i\le k}I(p_{i}\le s)$, it holds that 
\begin{align*}
\mathrm{FDP} &= \frac{V(s, \hat{k})}{R(s, \hat{k})\vee 1} = \frac{V(s, \hat{k})}{1 + \sum_{i\le \hat{k}}I(p_{i} > \lambda)}\cdot \frac{1 + \sum_{i\le \hat{k}}I(p_{i} > \lambda)}{R(s, \hat{k})\vee 1}\\
& \le \frac{V(s, \hat{k})}{1 + \sum_{i\in \nullset, i\le \hat{k}}I(p_{i} > \lambda)}\cdot \frac{1 + \sum_{i\le \hat{k}}I(p_{i} > \lambda)}{R(s, \hat{k})\vee 1}\\
& \le \frac{V(s, \hat{k})}{1 + \sum_{i\in \nullset, i\le \hat{k}}I(p_{i} > \lambda)}\cdot \frac{1 - \lambda}{s}q.
\end{align*}
Let 
\[M(k) = \frac{V(s, k)}{1 + \sum_{i\in \nullset, i\le k}I(p_{i} > \lambda)}.\]
Now we prove that $M(k)$ is a backward martingale with respect to the filtration $\{\mc{F}_{k}: k = n, n - 1, \ldots, 1\}$. In fact, let
\[V^{+}(k) = V(s, k) = \sum_{i\in \nullset, i\le k}I(p_{i}\le s), \quad V^{-}(k) = \sum_{i\in \nullset, i\le k}I(p_{i} > \lambda).\]
The notations here is comparable to \citeA{barber15}. Then 
\[M(k) = \frac{V^{+}(k)}{1 + V^{-}(k)}.\]
If $n\in \nullset^{c}$ is non-null, then $M(k - 1) = M(k)$. If $n$ is null, let $(I_{1}, I_{2}) = (I(p_{k} \le s), I(p_{k} > \lambda))$. Since $\{p_{i}: i\in \nullset\}$ are i.i.d., by symmetry, 
\[P(I_{1} = 1, I_{2} = 0 | \mc{F}_{k}) = \frac{V^{+}(k)}{|\nullset\cap \{1, \ldots, k\}|}, \quad P(I_{1} = 0, I_{2} = 1 | \mc{F}_{k}) = \frac{V^{-}(k)}{|\nullset\cap \{1, \ldots, k\}|},\]
\[P(I_{1} = 0, I_{2} = 0 | \mc{F}_{k}) = 1 - \frac{V^{+}(k) + V^{-}(k)}{|\nullset\cap \{1, \ldots, k\}|}.\]
Thus,
\begin{align*}
&\E (M(k - 1)|\mc{F}_{k})\\ 
= &\frac{V^{+}(k) - 1}{1 + V^{-}(k)}\cdot\frac{V^{+}(k)}{|\nullset\cap \{1, \ldots, k\}|} + \frac{V^{+}(k)}{V^{-}(k)\vee 1}\cdot \frac{V^{-}(k)}{|\nullset\cap \{1, \ldots, k\}|} + \frac{V^{+}(k)}{1 + V^{-}(k)}\cdot \lb 1 - \frac{V^{+}(k) + V^{-}(k)}{|\nullset\cap \{1, \ldots, k\}|}\rb\\
\le &\frac{V^{+}(k) - 1}{1 + V^{-}(k)}\cdot\frac{V^{+}(k)}{|\nullset\cap \{1, \ldots, k\}|} + \frac{V^{+}(k)}{|\nullset\cap \{1, \ldots, k\}|} + \frac{V^{+}(k)}{1 + V^{-}(k)}\cdot \lb 1 - \frac{V^{+}(k) + V^{-}(k)}{|\nullset\cap \{1, \ldots, k\}|}\rb\\
= & \frac{V^{+}(k)}{1 + V^{-}(k)}\cdot \frac{V^{+}(k) + V^{-}(k)}{|\nullset\cap \{1, \ldots, k\}|} + \frac{V^{+}(k)}{1 + V^{-}(k)}\cdot \lb 1 - \frac{V^{+}(k) + V^{-}(k)}{|\nullset\cap \{1, \ldots, k\}|}\rb\\
= & \frac{V^{+}(k)}{1 + V^{-}(k)} = M(k).
\end{align*}
In summary,
\[\E (M(k - 1) | \mc{F}_{k})\le M(k)\]
which shows that $M_{k}$ is a backward super-martingale. Notice that $M(k)\le n$ is bounded, it follows from optimal stopping time theorem that
\[\E M(\hat{k})\le \E M(n) = \E\lb\frac{\sum_{i\in \nullset}I(p_{i}\le s)}{1 + \sum_{i\in \nullset}I(p_{i} > \lambda)}\rb.\]
It is easy to see that 
\[\mc{L}\lb \sum_{i\in \nullset}I(p_{i}\le s) \bigg| \sum_{i\in \nullset}I(p_{i}\le \lambda) = m\rb = Binom\lb m, \frac{F_{1}(s)}{F_{1}(\lambda)}\rb\]
and hence
\begin{align*}
\E\lb\frac{\sum_{i\in \nullset}I(p_{i}\le s)}{1 + \sum_{i\in \nullset}I(p_{i} > \lambda)}
\rb & = \E\lb\E\lb\frac{\sum_{i\in \nullset}I(p_{i}\le s)}{1 + \sum_{i\in \nullset}I(p_{i} > \lambda)}\bigg| \sum_{i\in \nullset}I(p_{i}\le \lambda)\rb\rb\\
& = \frac{F_{1}(s)}{F_{1}(\lambda)}\cdot \E\lb\frac{\sum_{i\in \nullset}I(p_{i}\le \lambda)}{1 + \sum_{i\in \nullset}I(p_{i} > \lambda)}\rb\\
& \le \frac{F_{1}(s)}{1 - F_{1}(\lambda)} 
\end{align*}
where the last assertion follows from the fact that for any binomial random variable $X\sim N(r, p)$,
\begin{equation}\label{eq: mean_end}
\E \frac{X}{r + 1 - X} \le \frac{s}{1 - \lambda}.
\end{equation}
In fact, 
\begin{align*}
\E \frac{X}{r + 1 - X} & = \sum_{i=0}^{n}\frac{i}{r + 1 - i}\cdot \com{r}{i}p^{i}(1 - p)^{r - i}\\
& = \sum_{i=1}^{r}\frac{r!}{(i - 1)!(r + 1 - i)!}\cdot p^{i}(1 - p)^{r - i}\\
& = \frac{p}{1 - p}\cdot\sum_{i=0}^{r - 1}\com{r}{i}p^{i}(1 - p)^{r - i}\\
& \le \frac{p}{1 - p}.
\end{align*}
Since $p_{i}\succeq U[0, 1]$, it holds that $F_{1}(s)\le s$ and $F_{1}(\lambda)\le \lambda$. Thus, by Optional Stopping Theorem, 
\[\E M(\hat{k})\le \E M(n) = \E\lb\frac{\sum_{i\in \nullset}I(p_{i}\le s)}{1 + \sum_{i\in \nullset}I(p_{i} > \lambda)}\rb \le \frac{s}{1 - \lambda}.\] 
and hence
\[FDR(\hat{k}) = \E \lb\frac{V(s, \hat{k})}{R(s, \hat{k})\vee 1}\rb \le \frac{s}{1 - \lambda}\cdot \frac{1 - \lambda}{s}q = q.\]

\section*{Appendix B: Proof of Theorem \ref{thm:power_as}}
\begin{lemma}\label{lem: sum_ber}
  Let $B_{i}\sim Ber(1, p_{i})$ are independent Bernoulli random variables. Then for some positive integer $r$ and postive real number $\nu$,
\[P\lb\sup_{k\ge r}\bigabs{\frac{\sum_{i=1}^{k}(B_{i} - p_{i})}{k}} > \nu\rb\le \lb 2 + \frac{4}{\nu^{2}}\rb e^{-\frac{r\nu^{2}}{2}}.\]
\end{lemma}
\begin{proof}
Notice that $B_{i} - p_{i}$ is subgaussian with parameter $1$, we have
\[P\lb \bigabs{\frac{\sum_{i = 1}^{k}(B_{i} - p_{i})}{k}} > \nu\rb \le 2e^{-\frac{k\nu^{2}}{2}}.\]
Then 
\[P\lb\sup_{k\ge r}\bigabs{\frac{\sum_{i=1}^{k}(B_{i} - p_{i})}{k}} > \nu\rb\le 2\sum_{k\ge r}e^{-\frac{k\nu^{2}}{2}}\le \frac{2e^{-\frac{r\nu^{2}}{2}}}{1 - e^{-\frac{\nu^{2}}{2}}}\le \lb 2 + \frac{4}{\nu^{2}}\rb e^{-\frac{r\nu^{2}}{2}},\]
where the last step uses the fact that 
\[1 - e^{-\frac{\nu^{2}}{2}} = 1 - \frac{1}{e^{\frac{\nu^{2}}{2}}}\ge 1 - \frac{1}{1 + \frac{\nu^{2}}{2}} = \frac{1}{1 + \frac{2}{\nu^{2}}}.\]
\end{proof}

\begin{proposition}\label{prop:Pi}
For a VCT model with instantaneous non-null probability $\pi(t)$, 
\begin{equation}\label{eq:prop_non_null}
\max_{k = a_{n}, a_{n} + 1, \ldots, n}\bigg|\frac{\#\{i\le k: i\not\in \mathcal{H}_{0}\}}{k} - \Pi\lb\frac{k}{n}\rb\bigg|\le \eps_{n}
\end{equation}
holds with probability converging to 1 for properly chosen sequences $\{a_{n}\}$ and $\{\eps_{n}\}$ such that
\[a_{n} / n \rightarrow 0, \,\, a_{n}\rightarrow \infty, \,\, \eps_{n}\rightarrow 0.\]
In particular we can set $a_{n} = \lceil(\log n)^{2}\rceil$ and $\eps_{n} = \frac{1}{\sqrt{\log n}}$.
\end{proposition}
\begin{remark}
The condition (12) of \citeA{Li15} sets $a_{n} = 0$, which cannot be true. As will be shown later, a growing sequence $a_{n}$ suffices for our asymptotic analysis.
\end{remark}
\begin{proof}[\textbf{Proof of Proposition \ref{prop:Pi}}]
Let $P_{n}(x)$ be the step function with $P_{n}(x) = \frac{\lfloor nx\rfloor}{n}$. For any given $k$,
\begin{align*}
&\bigabs{\frac{\#\{i\le k: i\not\in \nullset\}}{k} - \Pi\lb\frac{k}{n}\rb}\\
\le & \bigabs{\frac{\sum_{i = 1}^{k}(I(i\not\in \nullset) - \pi(\frac{i}{n}))}{k}} + \bigabs{\int_{0}^{\frac{k}{n}}\pi(x)d(P_{n}(x) - x)}\\
\le & \bigabs{\frac{\sum_{i = 1}^{k}(I(i\not\in \nullset) - \pi(\frac{i}{n}))}{k}} + \frac{1}{n}\cdot \int_{0}^{\frac{k}{n}}\pi(x)dx\\
\le & \bigabs{\frac{\sum_{i = 1}^{k}(I(i\not\in \nullset) - \pi(\frac{i}{n}))}{k}} + \frac{1}{n}.
\end{align*}
Let $\nu_{n} = \eps_{n} - \frac{1}{n}$. By Lemma \ref{lem: sum_ber},
\[P\lb\sup_{k\ge a_{n}}\bigabs{\frac{\sum_{i=1}^{k}(I(i\not\in\nullset) - \pi(\frac{i}{n}))}{k}} > \nu_{n}\rb\le \lb 2 + \frac{4}{\nu_{n}^{2}}e^{-\frac{a_{n}\nu_{n}^{2}}{2}}\rb \triangleq \delta_{n}\rightarrow 0.\]
Thus with probability $1 - \delta_{n}$,
\[\sup_{k\ge a_{n}}\bigabs{\frac{\#\{i\le k: i\not\in \nullset\}}{k} - \Pi\lb\frac{k}{n}\rb} \le \eps_{n}.\]
\end{proof}

\begin{proof}[\textbf{Proof of Theorem \ref{thm:power_as}}]
Note that when $F_{1}$ is strictly concave,
\[\frac{1 - F_{1}(\lambda)}{1 - \lambda} = \frac{F_{1}(1) - F_{1}(\lambda)}{1 - \lambda}\le \frac{F_{1}(1) - F_{1}(0)}{1 - 0} = 1\le \frac{F_{1}(s) - F_{1}(0)}{s - 0} = \frac{F_{1}(s)}{s}.\]
We will use this result throughout the proof. Select $a_{n} = \lceil(\log n)^{2}\rceil$ and $\eps_{n} = \frac{1}{\sqrt{\log n}}$ and let
\[\delta_{n} \triangleq\lb 2 + \frac{4}{\nu_{n}^{2}}\rb\cdot e^{-\frac{a_{n}\nu_{n}^{2}}{2}}.\]
Then $\delta_{n}\rightarrow 0$ and by Lemma \ref{lem: sum_ber}, with probability $1 - 2\delta_{n}$,
\[\sup_{k\ge a_{n}}\bigabs{\frac{\sum_{i=1}^{k}I(p_{i} > \lambda)}{k} - \frac{\sum_{i=1}^{k}\E I(p_{i} > \lambda)}{k}}\le \nu_{n},\]
and 
\[\sup_{k\ge a_{n}}\bigabs{\frac{\sum_{i=1}^{k}I(p_{i} \le s)}{k} - \frac{\sum_{i=1}^{k}\E I(p_{i} \le s)}{k}}\le \nu_{n}.\]
On the other hand, by Proposition \ref{prop:Pi}, 
\begin{align*}
  &\bigabs{\frac{\sum_{i=1}^{k}\E I(p_{i} > \lambda)}{k} - \lb \left[ 1 - \Pi\lb\frac{k}{n}\rb\right](1 - \lambda) + \Pi\lb\frac{k}{n}\rb(1 - F_{1}(\lambda))\rb}\\
= & \bigabs{\frac{\#\{i\le k: i\not\in\nullset\}}{k} - \Pi\lb\frac{k}{n}\rb}\cdot|F_{1}(\lambda) - \lambda| \le \eps_{n}.
\end{align*}
Similarly, 
\[\bigabs{\frac{\sum_{i=1}^{k}\E I(p_{i} \le s)}{k} - \lb \left[ 1 - \Pi\lb\frac{k}{n}\rb\right]s + \Pi\lb\frac{k}{n}\rb F_{1}(s)\rb}\le \eps_{n}.\]
These imply that with probability $1 - 2\delta_{n}$, it holds uniformly for $k\ge a_{n}$ that
\begin{align*}
&\widehat{\mathrm{FDP}}_{AS}(k) \le \frac{\left[ 1 - \Pi\lb\frac{k}{n}\rb\right](1 - \lambda) + \Pi\lb\frac{k}{n}\rb(1 - F_{1}(\lambda)) + \eps_{n} + \nu_{n} + \frac{1}{a_{n}}}{\left[ 1 - \Pi\lb\frac{k}{n}\rb\right]s + \Pi\lb\frac{k}{n}\rb F_{1}(s) - \eps_{n} - \nu_{n}}\\ 
\Longrightarrow & \frac{1 + \sum_{i=1}^{k}I(p_{i} > \lambda)}{1\vee \sum_{i =1}^{k}I(p_{i}\le s)} - \frac{\left[ 1 - \Pi\lb\frac{k}{n}\rb\right](1 - \lambda) + \Pi\lb\frac{k}{n}\rb(1 - F_{1}(\lambda))}{\left[ 1 - \Pi\lb\frac{k}{n}\rb\right]s + \Pi\lb\frac{k}{n}\rb F_{1}(s)}\\
&\le \lb\eps_{n} + \nu_{n} + \frac{1}{a_{n}}\rb\cdot \frac{\left[1 - \Pi\lb\frac{k}{n}\rb\right](1 - \lambda) + \Pi\lb\frac{k}{n}\rb(1 - F_{1}(\lambda)) + \left[ 1 - \Pi\lb\frac{k}{n}\rb\right]s + \Pi\lb\frac{k}{n}\rb F_{1}(s)}{\lb\left[ 1 - \Pi\lb\frac{k}{n}\rb\right]s + \Pi\lb\frac{k}{n}\rb F_{1}(s) - \eps_{n} - \nu_{n}\rb^{2}}\\
& \le \lb\eps_{n} + \nu_{n} + \frac{1}{a_{n}}\rb\cdot \frac{1 - F_{1}(\lambda) + F_{1}(s)}{(s - \eps_{n} - \nu_{n})^{2}}  
\end{align*}
and 
\begin{align*}
&\widehat{\mathrm{FDP}}_{AS}(k)  \ge \frac{\left[ 1 - \Pi\lb\frac{k}{n}\rb\right](1 - \lambda) + \Pi\lb\frac{k}{n}\rb(1 - F_{1}(\lambda)) - \eps_{n} - \nu_{n} + \frac{1}{a_{n}}}{\left[ 1 - \Pi\lb\frac{k}{n}\rb\right]s + \Pi\lb\frac{k}{n}\rb F_{1}(s) + \eps_{n} + \nu_{n}}\\
\Longrightarrow & \frac{1 + \sum_{i=1}^{k}I(p_{i} > \lambda)}{1\vee \sum_{i =1}^{k}I(p_{i}\le s)} - \frac{\left[ 1 - \Pi\lb\frac{k}{n}\rb\right](1 - \lambda) + \Pi\lb\frac{k}{n}\rb(1 - F_{1}(\lambda))}{\left[ 1 - \Pi\lb\frac{k}{n}\rb\right]s + \Pi\lb\frac{k}{n}\rb F_{1}(s)}\\
&\ge -\lb\eps_{n} + \nu_{n}\rb\cdot \frac{\left[1 - \Pi\lb\frac{k}{n}\rb\right](1 - \lambda) + \Pi\lb\frac{k}{n}\rb(1 - F_{1}(\lambda)) + \left[ 1 - \Pi\lb\frac{k}{n}\rb\right]s + \Pi\lb\frac{k}{n}\rb F_{1}(s)}{\lb\left[ 1 - \Pi\lb\frac{k}{n}\rb\right]s + \Pi\lb\frac{k}{n}\rb F_{1}(s)\rb^{2}}\\
& \ge -\lb\eps_{n} + \nu_{n}\rb\cdot \frac{1 - F_{1}(\lambda) + F_{1}(s)}{s^{2}}  
\end{align*}
Since $\eps_{n}, \nu_{n}\rightarrow 0$, we have
\begin{equation}\label{eq: uniform_bound_1}
\lim_{n\rightarrow\infty} \sup_{k\ge a_{n}}\bigabs{\widehat{\mathrm{FDP}}_{AS}(k) - \frac{\left[ 1 - \Pi\lb\frac{k}{n}\rb\right](1 - \lambda) + \Pi\lb\frac{k}{n}\rb(1 - F_{1}(\lambda))}{\left[ 1 - \Pi\lb\frac{k}{n}\rb\right]s + \Pi\lb\frac{k}{n}\rb F_{1}(s)}} = 0\quad a.s.
\end{equation}
Recall that 
\[\mathrm{FDR}^{*}_{AS}(t) = \frac{1 - \Pi(t) + \Pi(t)\frac{1 - F_{1}(\lambda)}{1 - \lambda}}{1 - \Pi(t) + \Pi(t)\frac{F_{1}(s)}{s}},\]
then 
\[\frac{d\mathrm{FDR}^{*}_{AS}(t)}{d\Pi(t)} = \frac{\frac{F_{1}(s)}{s} - \frac{1 - F_{1}(\lambda)}{1 - \lambda}}{\lb\frac{F_{1}(s)}{s} - 1\rb\lb 1 + (\frac{F_{1}(s)}{s} - 1)\Pi(t)\rb}\le \frac{\frac{F_{1}(s)}{s} - \frac{1 - F_{1}(\lambda)}{1 - \lambda}}{\frac{F_{1}(s)}{s} - 1}\triangleq L\]
and hence $\mathrm{FDR}^{*}_{AS}(t)$ is $LL_{\Pi}$Lipschitz where $L_{\Pi}$ is the Lipschitz constant of $\Pi$. This entails that 
\begin{equation}\label{eq: lipschitz}
 \sup_{k\le n}\sup_{|t - \frac{k}{n}| < \frac{1}{n}}\bigabs{\mathrm{FDR}^{*}_{AS}(t) - \mathrm{FDR}^{*}_{AS}\lb\frac{k}{n}\rb}\le \frac{LL_{\Pi}}{n}.
\end{equation}
(\ref{eq: lipschitz}) together with (\ref{eq: uniform_bound_1}) implies that 
\begin{equation}\label{eq: uniform_bound_2}
\lim_{n\rightarrow \infty}\sup_{t\ge a_{n} / n}|\widehat{\mathrm{FDP}}_{AS}(\lfloor nt\rfloor) - \mathrm{FDR}^{*}_{AS}(t)| = 0\quad a.s.
\end{equation}
Since $a_{n} / n\rightarrow 0$, for any $c > 0$, 
\begin{equation}\label{eq: uniform_bound_3}
\lim_{n\rightarrow \infty}\sup_{t\ge c}|\widehat{\mathrm{FDP}}_{AS}(\lfloor nt\rfloor) - \mathrm{FDR}^{*}_{AS}(t)| = 0\quad a.s.
\end{equation}

\noindent If $\mathrm{FDR}^{*}_{AS}(0)\ge q$, then for any $c > 0$ and $x\ge c$, $\mathrm{FDR}^{*}_{AS}(t) \ge \mathrm{FDR}^{*}_{AS}(c) > \mathrm{FDR}^{*}_{AS}(0)\ge q$. (\ref{eq: uniform_bound_3}) implies that 
\[\liminf_{n\rightarrow\infty}\inf_{t\ge c}\widehat{\mathrm{FDP}}_{AS}(\lfloor nt\rfloor)\ge \mathrm{FDR}^{*}_{AS}(c) > q\quad a.s.\]
By definition,
\[\widehat{\mathrm{FDP}}_{AS}\lb\hat{k}_{AS}\rb \le q\]
and hence $\hat{k} / n\le c$ almost surely. This holds for arbitrary $c > 0$, therefore, $\hat{k} / n\stackrel{a.s.}{\rightarrow} 0 = t^{*}_{AS}$. In this case, 
\[\mathrm{Pow}_{AS} = \frac{\# \{i\le \hat{k}: i\not\in \nullset, p_{i}\le s\}}{\#\{i\le n: i\not\in \nullset\}}\le \frac{\hat{k}}{n}\cdot \frac{n}{\#\{i\le n: i\not\in \nullset\}}\stackrel{a.s.}{\rightarrow}0\]
since $\hat{k} / n\stackrel{a.s.}{}b\rightarrow 0$ and 
\[\frac{n}{\#\{i\le n: i\not\in \nullset\}} \stackrel{a.s.}{\rightarrow} \frac{1}{\Pi(1)} < \infty.\]

~\\

\noindent If $\mathrm{FDR}^{*}_{AS}(1) \le q$, similar to the above argument, we have
\[\limsup_{n\rightarrow\infty}\widehat{\mathrm{FDP}}_{AS}(\lfloor n(1 - c)\rfloor)\le \mathrm{FDR}^{*}_{AS}(1 - c) < \mathrm{FDR}^{*}_{AS}(1) \le q\]
for arbitrary $c > 0$. This implies that  $\hat{k}/n \ge 1 - c$ almost surely. Thus, $\hat{k} / n\stackrel{a.s.}{\rightarrow} 1 = t^{*}_{AS}$. In this case,
\begin{align}
&\mathrm{Pow}_{AS} = \frac{\# \{i\le \hat{k}: i\not\in \nullset, p_{i}\le s\}}{\#\{i\le n: i\not\in \nullset\}}\nonumber\\ 
= &\frac{\# \{i\le \hat{k}: i\not\in \nullset, p_{i}\le s\}}{\#\{i\le \hat{k}: i\not\in \nullset\}}\cdot\frac{\#\{i\le \hat{k}: i\not\in \nullset\}}{\hat{k}}\cdot \frac{\hat{k}}{n}\cdot \frac{n}{\#\{i\le n: i\not\in \nullset\}}\label{eq: power_decompose}
\end{align}
Since $\hat{k} / n\stackrel{a.s.}{\rightarrow} 1 > 0$, $\hat{k} \ge a_{n}$ almost surely and hence
\[\Pi\lb\frac{\hat{k}}{n}\rb - \eps_{n} \ge \frac{\#\{i\le \hat{k}: i\not\in \nullset\}}{\hat{k}} \ge \Pi\lb\frac{\hat{k}}{n}\rb - \eps_{n}.\]
This implies that
\[\frac{\#\{i\le \hat{k}: i\not\in \nullset\}}{\hat{k}}\rightarrow \Pi(1)\quad a.s.\]
and as a byproduct, we know $\#\{i\le \hat{k}: i\not\in \nullset\}\stackrel{a.s.}{\rightarrow}\infty$. By Law of Large Number,
\[\frac{\# \{i\le \hat{k}: i\not\in \nullset, p_{i}\le s\}}{\#\{i\le \hat{k}: i\not\in \nullset\}}\stackrel{a.s.}{\rightarrow}F_{1}(s).\]
Therefore, 
\[\mathrm{Pow}_{AS}\rightarrow \Pi(1)\cdot F_{1}(s)\cdot 1 \cdot \frac{1}{\Pi(1)} = F_{1}(s)\quad a.s.\]

~\\

\noindent If $\mathrm{FDR}^{*}_{AS}(0) < q < \mathrm{FDR}^{*}_{AS}(1)$, then $t^{*}_{AS} = \mathrm{FDR}^{*-1}_{AS}(q)$. For any $c > 0$, (\ref{eq: uniform_bound_3}) implies that 
\[\limsup_{n\rightarrow \infty}\widehat{\mathrm{FDP}}_{AS}(\lfloor n(t^{*}_{AS} - c)\rfloor)\le \mathrm{FDR}^{*}_{AS}(t^{*}_{AS} - c) < \mathrm{FDR}^{*}_{AS}(t^{*}_{AS}) = q,\]
and 
\[\liminf_{n\rightarrow\infty}\sup_{t\ge t^{*}_{AS} + c}\widehat{\mathrm{FDP}}_{AS}(\lfloor nt\rfloor)\ge \mathrm{FDR}^{*}_{AS}(t^{*}_{AS} + c) > \mathrm{FDR}^{*}_{AS}(t^{*}_{AS}) = q.\]
Thus,
\[t^{*}_{AS} - c\le \frac{\hat{k}}{n}\le t^{*}_{AS} + c\quad a.s.\]
Since $c$ is arbitrary, we have $\hat{k} / n\stackrel{a.s.}{\rightarrow}t^{*}_{AS}$. In this case, notice that $\liminf_{n\rightarrow\infty}\hat{k} / n > 0$, we can apply the same argument as above and it follows from (\ref{eq: power_decompose}) that
\[\mathrm{Pow}_{AS}\rightarrow \Pi(t^{*}_{AS})\cdot F_{1}(s) \cdot t^{*}_{AS}\cdot \frac{1}{\Pi(1)} = \frac{t^{*}_{AS}\Pi(t^{*}_{AS})F_{1}(s)}{\Pi(1)}.\]
\end{proof}

\bibliography{power_of_ordered_test_icml2016_long}
\bibliographystyle{apacite}

\end{document}